\documentclass[11pt]{article}

\usepackage{amsmath,amssymb}
\usepackage{fullpage}
\usepackage[english]{babel}
\usepackage[T1]{fontenc}
\usepackage[ruled,noend,linesnumbered]{algorithm2e}
\usepackage{xspace}

\usepackage{wrapfig}

\usepackage{graphicx}
\graphicspath{ {./img/} }

\usepackage{hyperref}

\usepackage{amsthm}
\theoremstyle{plain}
\newtheorem{theorem}{Theorem}[section]
\newtheorem{lemma}[theorem]{Lemma}

\theoremstyle{definition}

\theoremstyle{remark}

\newcommand{\R}{\ensuremath{\mathbb{R}}\xspace}
\newcommand{\Ex}{\ensuremath{\mathbf{E}}\xspace}
\newcommand{\setS}{\ensuremath{S}\xspace}
\newcommand{\Part}{\ensuremath{\mathrm{Part}}\xspace}
\newcommand{\uu}{\ensuremath{\mathbf{u}}\xspace}

\newcommand{\etal}{\emph{et al.}\xspace}

\DeclareMathOperator{\conv}{conv}

\DeclareMathOperator{\VD}{VD}

\usepackage{authblk}

\title{Time-Space Trade-offs for Triangulations and Voronoi Diagrams}
\author[1]{Matias Korman}
\author[,2]{Wolfgang Mulzer\thanks{WM and PS were supported in part by 
DFG Grants MU 3501/1 and MU 3501/2. WM was supported in part by
ERC StG 757609. YS was supported by the DFG within 
the research training group ``Methods for Discrete Structures'' (GRK
    1408).}}
\author[3,4]{Andr\'e van Renssen}
\author[3,4]{Marcel Roeloffzen}
\author[,2]{Paul Seiferth\protect\footnotemark[1]}
\author[,2]{Yannik Stein\protect\footnotemark[1]}

\affil[1]{Tohoku University, Sendai, Japan.\\
    \texttt{mati@dais.is.tohoku.ac.jp}}
\affil[2]{Institut f\"ur Informatik, Freie Universit\"at Berlin, Germany.
    \texttt{\{mulzer,pseiferth,yannikstein\}@inf.fu-berlin.de}}
\affil[3]{National Institute of Informatics (NII), Tokyo, Japan.\\
    \texttt{\{andre,marcel\}@nii.ac.jp}}
\affil[4]{JST, ERATO, Kawarabayashi Large Graph Project.}
\date{}

\begin{document}

\maketitle

\begin{abstract}
Let $S$ be a planar $n$-point set.
A \emph{triangulation} for $S$ is a maximal plane
straight-line graph with vertex set $S$.
The \emph{Voronoi diagram} for $S$ is the subdivision
of the plane into cells such that all points in a cell
have the same nearest neighbor in $S$.
Classically, both structures can be computed
in $O(n \log n)$ time and $O(n)$ space.
We study the situation when the available workspace is
limited: given a parameter $s \in \{1, \dots, n\}$,
an $s$-workspace algorithm has read-only access
to an input array with the points from $S$ in arbitrary order,
and it may use only
$O(s)$ additional words of $\Theta(\log n)$ bits for reading and writing
intermediate data. The output should then be written to a write-only structure.
We describe a deterministic $s$-workspace algorithm
for computing an arbitrary triangulation of $S$ in time
$O(n^2/s + n \log n  \log s )$ and a randomized
$s$-workspace algorithm for finding the Voronoi diagram
of $S$ in expected time
$O((n^2/s) \log s + n \log s \log^*s)$.
\end{abstract}

\section{Introduction}

Since the early days of computer science, a major concern
has been to cope with
strong memory constraints. This
started in the
'70s~\cite{p-amsacm-69} when memory was
expensive. Nowadays, a major motivation comes from a
proliferation of small embedded devices
where large memory is neither feasible nor desirable (e.g.,
due to constraints on budget, power, size, or simply to make the device less attractive to thieves).

Even when memory size is not an issue, we might
want to limit the number of write operations:
one can read flash
memory quickly, but writing (or even reordering) data is slow
and may reduce the lifetime of the storage system; write-access
to removable memory may be limited for
technical or security reasons (e.g., when using
read-only media such as DVDs or to prevent leaking
information about the algorithm). Similar problems occur
when concurrent algorithms
access data simultaneously.
A natural way to address this is to consider
algorithms that do not modify the input.

The exact setting may vary,
but there is a common theme: the input resides in
read-only memory, the output must be written
to a write-only structure, and we can use
$O(s)$ additional variables to find the solution (for a
parameter $s$). The goal is to design algorithms
whose running time decreases as $s$ grows, giving a
{\em time-space trade-off}~\cite{s-mcepc-08}.
One of the first problems considered in this model
is {\em sorting}~\cite{mp-ssls-80,mr-sromswmdm-96}.
Here, the time-space product is known to be
$\Omega(n^2)$~\cite{bc-atstosgsmc-82}, and matching upper bounds
for the case $b\in \Omega(\log n) \cap O(n/\log n)$ were
obtained by Pagter and Rauhe~\cite{pr-otstofs-98} ($b$ denotes
the available workspace in \emph{bits}).

Our current notion of memory constrained algorithms was introduced to
computational geometry by Asano~\etal~\cite{amrw-cwagp-10},
who showed how to compute many classic geometric
structures with $O(1)$ workspace (related models were studied before~\cite{BronnimanChCh04}).
Later, time-space trade-offs were given
for problems on simple polygons, e.g., shortest
paths~\cite{abbkmrs-mcasp-11}, visibility~\cite{bkls-cvpufv-13},
or the convex hull of the vertices~\cite{bklss-sttosba-14}.

We consider a model in which the set $S$ of $n$ points is in an array
such that random access to each input point is possible, 
but we may not change or even reorder the input.
Additionally, we have
$O(s)$ variables (for a parameter $s \in \{1, \dots,  n\}$).
We assume that each variable or pointer contains a data word of $\Theta(\log n)$ bits.
Other than this, the model allows the usual word RAM operations.
In this setting we study two problems: computing an 
arbitrary triangulation for $S$ and
computing the Voronoi diagram $\VD(S)$ for $S$.
Since the output cannot be stored explicitly, the
goal is to report the edges of the triangulation or the vertices
of $\VD(S)$ successively, in no particular order.
Dually, the latter goal may be phrased in
terms of Delaunay triangulations.
We focus on  Voronoi diagrams, as
they lead to a more natural presentation.

Both problems can be solved in $O(n^2)$ time with $O(1)$
workspace~\cite{amrw-cwagp-10} or in
$O(n \log n)$ time with $O(n)$ workspace~\cite{bcko-cgaa-08}. However,
to the best of our knowledge, no trade-offs were known before.
Our triangulation algorithm achieves a running time of
$O(n^2/s + n \log n \log s )$ using $O(s)$ variables.
A key ingredient is the recent time-space trade-off
by Asano and Kirkpatrick for triangulating a special type of simple
polygons~\cite{ak-tstanlnp-13}. This also lets us obtain
significantly better running times for the case that the input is
sorted in $x$-order; see Section~\ref{sec_triang}.
For Voronoi diagrams, we use
random sampling to find the result in expected time
$O((n^2\log s)/s + n\log s \log^* s))$; see Section~\ref{sec_CS}.
Together
with recent work of Har-Peled~\cite{Har-Peled16}, this
appears to be one of the first uses
of random sampling to obtain space-time trade-offs for geometric
algorithms.
The sorting lower
bounds also apply to triangulations and Voronoi
diagrams (since we
can reduce the former to the latter). 
This implies that our second algorithm 
is almost optimal.

\section{A Time-Space Trade-Off for General Triangulations}
\label{sec_triang}

In this section we describe an algorithm that outputs the edges of
a triangulation for a given point set $S$ in arbitrary order.
For ease in the presentation we first assume that $S$ is presented in sorted order.
In this case, a time-space trade-off follows quite readily from
known results. We then show how to generalize this for arbitrary inputs,
which requires a careful adaptation of the existing data structures.

\subsection{Sorted Input}
\label{sec_sorted}

Suppose the input points
$\setS = \{q_1, \ldots, q_n\}$ are stored in increasing order of
$x$-coordinate and that all $x$-coordinates are distinct,
i.e., $x_i < x_{i+1}$ for $1 \leq i < n$, where $x_i$ denotes
the $x$-coordinate of $q_i$. 

A crucial ingredient in our algorithm is a recent result
by Asano and Kirkpatrick for triangulating \emph{monotone mountains}\footnote{Also
known as \emph{unimonotone polygons}~\cite{FournierMo84}.} (or \emph{mountains} for short).
A mountain is a simple polygon
with vertex sequence $v_1, v_2, \dots, v_k$ such
that the $x$-coordinates of the vertices increase monotonically.
The edge $v_1v_k$ is called the \emph{base}.
Mountains can be triangulated very efficiently with bounded workspace.

\begin{theorem}[Lemma~3 in \cite{ak-tstanlnp-13}, rephrased]
  \label{obs:asanokirck}
  Let $H$ be a mountain with $n$ vertices, stored in sorted $x$-order
  in read-only memory. Let $s \in \{2, \dots, n\}$. We can report the 
  edges of a triangulation of $H$ in $O(n\log_s n )$ time and
  using $O(s)$ words of space.
\end{theorem}

Since $\setS$ is given in $x$-order, the edges $q_i q_{i+1}$, for $1 \leq i<n$, form a
monotone simple polygonal chain. Let $\Part(\setS)$ be the subdivision obtained by the
union of this chain with the edges of the convex hull of \setS\ (denoted by $\conv(\setS)$). 
A convex hull edge is {\em long} if the difference between
its indices is at least two (i.e., the endpoints are not consecutive).
The following lemma (illustrated in Fig.~\ref{fig:monotone}) lets us
decompose the problem into smaller pieces.

\begin{figure}[h]
  \centering
  \includegraphics[width=0.5\textwidth]{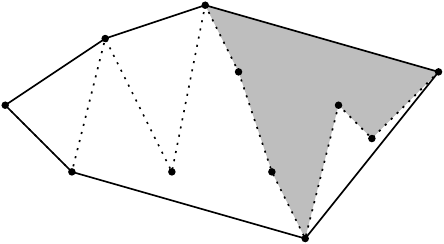}
  \caption{Any face of $\Part(\setS)$ is a mountain that is uniquely associated
    with a long convex hull edge.}
  \label{fig:monotone}
\end{figure}

\begin{lemma}\label{lemPartition} 
  Any bounded face of $\Part(\setS)$ is a mountain
  whose base is a long convex hull edge.
  Moreover, no point of $\setS$ lies in more than four faces 
  of $\Part(\setS)$.
\end{lemma}
\begin{proof} 
Any point $q_i \in \setS$ has at most four 
neighbors in $\Part(\setS)$: $q_{i-1}$, $q_{i+1}$,
its predecessor and its successor along the convex hull (if $q_i$ lies on 
$\conv(\setS)$). Thus, no point of \setS belongs to more than
four faces of $\Part(\setS)$.

Next we show that every face $F$ of $\Part(\setS)$ 
is a mountain with a \emph{long} convex-hull edge as its base.
The boundary of $F$ contains at least one long convex-hull edge
$e = (q_i,q_j)$ ($i < j)$, as other 
edges connect only consecutive vertices. 
Since the monotone path $q_i, \dots, q_j$ forms a cycle with 
the edge $e$ and since the boundary of $F$ is a simple polygon, we 
conclude 
that $e$ is the only long convex-hull edge bounding $F$. Recall that 
$e$ is a convex hull edge,
and thus all points $q_{i+1}, \dots, q_{j-1}$ 
lie on one side of $e$ and form a monotone chain (and in particular 
$F$ is a mountain with base $e$). 
\end{proof}

The algorithm for sorted input is now very simple.
We  compute the edges of the convex hull (starting from the leftmost 
point and proceeding in clockwise order).
Whenever a long edge would be reported, we pause the convex hull 
algorithm, and we triangulate the corresponding mountain. Once the 
mountain has been triangulated, we resume with the convex hull 
algorithm until all convex hull edges have been computed. 
The trade-off now follows from already existing trade-offs in the
various subroutines. 

\begin{theorem}\label{theo_sorted}
Let $\setS$ be a set of $n$ points, sorted in $x$-order.
We can report the edges of a triangulation of $\setS$ in $O(n^2)$ time using $O(1)$
variables, in $O(n^2\log n/2^s)$ time using $O(s)$ variables
(for any $s\in \Omega(\log \log n) \cap o(\log n)$), and in
$O(n\log_p n)$ time using $O(p\log_p n)$ variables (for any
  $2\leq p\leq n$).
\end{theorem}

\begin{proof}
Correctness follows from 
Lemma~\ref{lemPartition}, so we
focus on the performance analysis. 
The main steps are: (i) computing the
convex hull of a point set given in $x$-order; 
and (ii) triangulating a mountain.

By Theorem~\ref{obs:asanokirck}, we can triangulate
a mountain $F_i$ with $n_i$ vertices in
time $O(n_i\log_s n_i)$ with $O(s)$ variables. 
We do not need to store $F_i$ explicitly, since
its vertices constitute a consecutive subsequence of $S$ and can
be specified by the two endpoints of the base.
No vertex appears in more than
four mountains by Lemma~\ref{lemPartition}, 
so the total time for  triangulating the mountains is 
$\sum_{i} O(n_i \log_s n_i) = O(n\log_s n)$.
By reusing space, we can ensure that the total space requirement 
is $O(s)$.

Now we bound the time for computing $\conv(\setS)$.
This algorithm is paused to triangulate mountains,
but overall it is executed only once.
There are several convex hull algorithms for sorted point
sets under memory constraints. If $s\in \Theta(1)$, we can use
gift-wrapping (Jarvis march~\cite{j-ich-73}), which
runs in $O(n^2)$ time. Barba~\etal~\cite{bklss-sttosba-14} provided a different
algorithm that runs in $O(n^2\log n/2^s)$ time
using $O(s)$ variables (for any $s\in o(\log n)$).\footnote{In fact, Barba~\etal show how
to compute the convex hull of a simple polygon, but also show that both problems are
equivalent. The monotone chain can be completed to a polygon by adding a vertex with a
very high or low $y$-coordinate.}
This approach is desirable for $s\in \Omega(\log\log n) \cap o(\log n)$. As soon as
$s = \Omega(\log n)$, we can use the approach of Chan and
Chen~\cite{cc-mpga-07}. This algorithm runs in $O(n\log_p n)$ time and uses $O(p\log_p n)$
variables, for any $2 \leq p \leq n$. Regardless of the size of the workspace, the time for computing the convex hull
dominates the time needed for triangulating all mountains.
\end{proof}

A similar approach is unlikely to work for the Delaunay triangulation,
since knowing the $x$-order of the
input does not help in computing it~\cite{dl-ocvdsps-95}.

\subsection{General Input}
The algorithm from Section~\ref{sec_sorted} 
uses the sorted order in two ways. Firstly, 
the convex-hull algorithms of Barba~\etal~\cite{bklss-sttosba-14} 
and of Chan and Chen~\cite{cc-mpga-07} work 
only for simple polygons (e.g., for sorted input).
Instead, we use the algorithm by Darwish and 
Elmasry~\cite{de-otst2dchp-14} that gives the 
upper (or lower) convex hull of any sequence 
of $n$ points in $O(n^2/(s\log n) + n \log n)$
time with $O(s)$ variables\footnote{Darwish and 
Elmasry~\cite{de-otst2dchp-14} state a running 
time of $O(n^2/s+n\log n)$, but they measure 
workspace in bits, while we use words.},
matching known lower bounds. Secondly, and 
more importantly, the Asano-Kirkpatrick (AK) 
algorithm for triangulating a mountain requires the input to be sorted.
To address this issue, we simulate sorted 
input using multiple heap structures. This 
requires a close examination of how the
AK-algorithm accesses its input.

Let $F$ be a mountain with $n$ vertices. Let 
$F^\uparrow$ and $F^\downarrow$ denote the 
vertices of $F$ in ascending and in 
descending $x$-order. The AK-algorithm has two phases, one 
focused on $F^\uparrow$ and the other one on 
$F^\downarrow$.\footnote{AK reduce triangulation 
to the \emph{next smaller right neighbor} (NSR) 
and the \emph{next smaller left neighbor} (NSL) 
problem and present an algorithm for NSR if 
the input is in $x$-order. This implies an 
NSL-algorithm by reading the input in reverse.} 
Each pass computes a portion of the triangulation edges,
uses $O(s)$ variables, and scans the input $\Theta(\log_s n)$ 
times. We focus on the approach for $F^\uparrow$. 

As mentioned, the algorithm uses $\Theta(\log_s n)$ rounds. 
In round $i$, it partitions $F$ into blocks of 
$O(|F|/s^i)$ consecutive points that are processed 
from left to right. Each block is further subdivided
into $O(s)$ sub-blocks $b_1,\dots,b_k$ of size 
$O(|F|/s^{i+1})$. The algorithm does two scans over
the sub-blocks. The first scan processes the elements 
from left to right. Whenever the first scan finishes reading 
a sub-block $b_i$, the algorithm makes $b_i$ 
\emph{active} and creates a pointer $l_i$ 
to the rightmost element of $b_i$. The second scan 
goes from right to left and is concurrent to the 
first scan. In each step, it reads the element at 
$l_i$ in the rightmost active sub-block $b_i$, and 
it decreases $l_i$ by one. If $l_i$ leaves $b_i$, then
$b_i$ becomes inactive. As the first scan creates new
active sub-blocks as it proceeds, the second scan may 
jump between sub-blocks. The interested reader may find
a more detailed description in~\ref{app:ak}.

To provide the input for the AK-algorithm, we need the heap by 
Asano~\etal~\cite{aek-pqsrod-13}. For completeness, we briefly restate its 
properties here.

\begin{lemma}[\cite{aek-pqsrod-13}]\label{lem:heap}
  Let $S$ be a set of $n$ points. There is a heap 
  that supports insert and extract-min (resp.\@
  extract-max) in $O\big((n/(s \log n) + \log s ) D(n))$ time using
  $O(s)$ variables, where $D(n)$ is the time to decide whether a given
  element currently resides in the heap (is \emph{alive}).\footnote{The 
  bounds in~\cite{aek-pqsrod-13} do not include the factor $D(n)$ since 
  the authors studied a setting similar to Lemma~\ref{lem:minheap} where 
  it takes $O(1)$ time to decide whether an element is alive.}
\end{lemma}

\begin{proof}
We first describe the data structure.
Then we discuss how to perform insertions and 
extract-min operations.

We partition the input into $s \log n$ consecutive 
\emph{buckets} of equal size, and we build a
complete binary tree $T$ over the buckets. Let $v$ be 
a node of $T$ with height $h$. Then, there are 
$2^h$ buckets below $v$ in $T$. We store $2h$
\emph{information bits} in $v$ to
specify the minimum alive element below $v$. The first 
$h$ bits identify the bucket containing the minimum. 
We further divide this bucket into $2^h$ consecutive
parts of equal size, called \emph{quantiles}. The 
second $h$ bits in $v$ specify the quantile containing 
the minimum. If $2h > \log n$, we use
$\log n$ bits to specify the minimum directly.
Hence, the total number of bits is bounded by
\begin{align*}
  \sum_{h=0}^{\log(s \log n)} 
  \frac{s \log n}{2^h} \min\{2h, \log n\} = 
  O(s \log n).
\end{align*}
Therefore we need $O(s)$ variables in total.

Let $v$ be a node with height $h$. To find the 
minimum alive element in $T$ below $v$, we use 
the $2h$ information bits stored in $v$. First,
we identify the bucket containing the minimum 
and the correct quantile within this bucket. 
This quantile contains $O\big(\frac{n}{2^h s \log n}\big)$ 
elements. For each element in the quantile, we 
decide in $D(n)$ time whether it is alive, 
and we return the minimum such element.
This takes $O\big(\frac{n}{2^h s \log n} D(n)\big)$ time in total.

\begin{description}

\item[insert:] Assume we want to insert an element 
$x$ that is at position $i$ in the input array. Let $v$ 
be the parent of the leaf of $T$ corresponding to the 
bucket that contains $x$. We update the information
bits at each node $u$ on the root path starting at $v$. 
To do so, we use the information bits in $u$ to find 
the minimum element in the buckets covered by $u$, as 
described above. Then we compare it with $x$. If $x$ is
larger, we are done and we stop the insertion. 
Otherwise, we update the information bits at $u$ 
to the bucket and quantile that contain
$x$. If we reach and update the root node, we also 
update the pointer that points to the minimum element 
in the heap. The work per node is dominated by the
costs for finding the minimum, which is 
$O\big(\frac{n}{2^h s \log n} D(n)\big)$.
Thus, the total cost for insertion is bounded by 
\begin{equation*}
 \sum_{h=0}^{\log(s \log n)} \frac{n}{2^h s \log n} D(n) = 
 O\Big (\frac{n}{s \log n} D(n)\Big).
\end{equation*}

\item[extract-min:]
First we use the pointer to the minimum alive element 
to determine the element $x$ to return. Then we use a 
similar update strategy as for insertions. Let $v$ be the 
leaf node corresponding to the bucket of $x$. We first 
update the information bits of $v$ by scanning through the whole 
bucket of $v$ and determining the smallest alive element. Since a bucket
contains $O(n/s\log n)$ elements, this needs time 
$O(n/(s \log n)D(n))$. Then we
update the information bits of each node $u$ on the path 
for $v$ as follows: let $v_1$ 
and $v_2$ be the two children of $u$. We determine the minimum 
alive element in the buckets covered by $v_1$ and $v_2$, take the smaller 
one, and use it to update the information bits at $u$. 
Once we reach the root, we also update the pointer to the minimum element
of the heap to the new minimum element of the root. The 
total time again is bounded by $O\big(\frac{n}{s \log n} D(n)\big )$.
\end{description}
\end{proof}

\begin{lemma}[\cite{aek-pqsrod-13}]\label{lem:minheap}
  Let $S$ be a set of $n$ points. We can build a heap with all
  elements in $S$ in $O(n)$ time that supports extract-min
  in $O\big(n/(s \log n) + \log n)$ time using $O(s)$ variables.
\end{lemma}
\begin{proof}
  The construction time is given in~\cite{aek-pqsrod-13}.
  To decide in $O(1)$ time if some $x \in S$ is
  alive, we store the last extracted minimum $m$ and test whether
  $x > m$.
\end{proof}

We now present the complete algorithm. We show how to
subdivide $S$ into mountains $F_i$ and how to run the AK-algorithm on each
$F_i^\uparrow$. By reversing the order, the same discussion applies to
$F_i^\downarrow$.
Sorted input is emulated by two heaps $H_1$, $H_2$ for $S$
according to $x$-order. By
Lemma~\ref{lem:minheap}, each heap uses $O(s)$ space, can
be constructed in $O(n)$ time, and supports extract-min in $O(n/(s\log
n) + \log n )$ worst-case time. We will use $H_1$ to determine the size of the
next mountain $F_i$ and $H_2$ to process the points of $F_i$.

We execute the convex hull algorithm with $\Theta(s)$ space until 
it reports the next convex hull edge $pq$.
Throughout the execution of the algorithm, heaps $H_1$ and $H_2$ contain exactly the points to
the right of $p$. 
We repeatedly extract the
minimum of $H_1$ until $q$ becomes the minimum element. Let $k$
be the number of removed points.

If $k = 1$, then $pq$ is short. We extract
the minimum of $H_2$, and we continue with the convex hull
algorithm.
If $k \geq 2$, then Lemma~\ref{lemPartition} shows that $pq$ is
the base of a mountain $F$ that consists of all points 
between $p$ and $q$. These are
exactly the $k+1$ smallest elements in $H_2$ (including $p$ and $q$). If
$k \leq s$, we extract them from $H_2$, and we triangulate $F$ in memory.
If $k > s$, we execute the AK-algorithm on $F$ using $O(s)$ variables. 
At the beginning of the $i$th round,
we create a copy $H_{(i)}$ of $H_2$, i.e., we duplicate the
$O(s)$ variables that determine the state of $H_2$. Further, we create 
an empty max-heap $H_{(ii)}$ using $O(s)$ variables to provide 
input for the second scan. To be able to reread a sub-block, we 
create a further copy $H'_{(i)}$ of
$H_2$.
Whenever the AK-algorithm requests the next point in the first scan,
we simply extract the minimum of $H_{(i)}$. When a sub-block is fully
read, we use $H'_{(i)}$ to reread the elements and insert
them into $H_{(ii)}$.  
Now, the rightmost element of all active sub-blocks corresponds 
exactly to the maximum of $H_{(ii)}$. 
One step in the second scan 
is equivalent to an extract-max on $H_{(ii)}$.

At the end of a round, we delete $H_{(i)}$, $H'_{(i)}$,
and $H_{(ii)}$, so that the space can be reused in the next round.
Once the AK-algorithm finishes, we repeatedly extract the minimum 
of $H_2$ until we reach $q$. 

\begin{theorem}\label{theo_unsorted} 
We can report the edges of a triangulation
of a set \setS of $n$ points in time $O(n^2/s + n\log n  \log s )$ 
using $O(s)$ additional variables.
\end{theorem}
\begin{proof}
Similarly as before, correctness directly 
follows from Lemma~\ref{lemPartition} and
the correctness of the AK-algorithm. The bound on the space usage is
immediate.

Computing the convex hull now needs
$O(n^2/(s\log n) + n \log n)$ time~\cite{de-otst2dchp-14}. By
Lemma~\ref{lem:minheap}, the heaps $H_1$ and $H_2$ can be constructed in
$O(n)$ time. During execution, we perform $n$ extract-min operations 
on each heap, requiring $O(n^2/(s \log n) + n \log n )$ time in total.

Let $F_j$ be a mountain with $n_j$ vertices that is 
discovered by the convex hull algorithm. 
If $n_j \leq s$, then $F_j$ is triangulated in memory in $O(n_j)$ time,
and  the total time for such mountains is $O(n)$.
If $n_j > s$, then the AK-algorithm runs in
$O(n_j \log_s n_j)$ time. We must also account for 
providing the input for the algorithm. For this, consider 
some round $i\geq 1$. We copy $H_2$ to $H_{(i)}$ in $O(s)$ time.
This time can be charged to the first scan, since $n_j > s$. 
Furthermore, we perform $n_j$ extract-min operations on 
$H_{(i)}$. Hence the total time to provide input for the 
first scan is $O(n_j n/(s \log n) + n_j \log n)$.

For the second scan, we create another copy $H'_{(i)}$ of $H_2$.
Again, the time for this can be charged to the scan.
Also, we perform $n_j$ extract-min operations
on $H'_{(i)}$ which takes $O(n_j n/(s \log n) + n_j \log n)$ time. 
Additionally, we insert each fully-read block into $H_{(ii)}$. 
The main problem is to determine if an element in $H_{(ii)}$ is 
alive: there are at most $O(s)$ active sub-blocks. For each 
active sub-block $b_i$, we know the first element $y_{i}$ and 
the element $z_i$ that $l_i$ points to. An element is
alive if and only if it is in the interval $[y_{i},z_{i}]$ for some 
active $b_i$. This can be checked in $O(\log s)$ time. Thus, by
Lemma~\ref{lem:heap}, each insert and extract-max on $H_{(ii)}$ 
takes $O\big((n/(s \log n) + \log s) \log s)$ time. Since each element
is inserted once, the total time to provide input to the second scan
is $O(n_j \log(s) (n /(s \log n) + \log s))$. This dominates
the time for the first scan.
 There are $O(\log_s n_j)$ rounds, so we can triangulate $F_j$ in time
  $O\big(n_j\log_s n_j +  n_j \log (n_j) \big( n /(s \log n) +
  \log s \big)\big)$.
  Summing over all $F_j$, the total time is
 $O(n^2/s + n\log n  \log s )$. 
\end{proof}

\section{Voronoi Diagrams}\label{sec_CS}

Given a planar $n$-point set $S$, we would like to find
the vertices of $\VD(S)$.
Let $K = \{p_1, p_2, p_3\}$ be a triangle
with $S \cap K = \emptyset$, $S \subseteq \conv(K)$, and so that
all vertices of $\VD(S)$ are vertices of $\VD(S \cup K)$.
For example, we can set $K = \{(-\kappa, -\kappa), 
(-\kappa, \kappa), (0, \kappa)\}$ for some large $\kappa > 0$. Since the
desired properties hold for all large enough $\kappa$,
we do not need to find an explicit value for it. Instead, whenever
we want to evaluate a predicate involving points from $K$, we can take the
result obtained for $\kappa \rightarrow \infty$.

Our algorithm relies on random sampling.
First, we show how to take
a random sample from $S$ with small workspace. One of many possible
approaches is the following one that ensures a 
worst-case guarantee:

\begin{lemma}\label{lem:sampling}
We can sample a uniform random subset $R \subseteq S$
of size $s$ in time $O(n + s \log s)$ and space $O(s)$.
\end{lemma}

\begin{proof}
The sampling algorithm consists of two phases.
In the first phase,
we sample a random sequence $I$ of $s$ distinct numbers from
$[n]$.\footnote{We write $[n]$ for the set $\{1, \dots, n\}$.}  
The phase proceeds in $s$ \emph{rounds}. 
At the beginning of round $k$, for $k = 1, \dots, s$, we have already sampled a
sequence $I$ of $k-1$ numbers from $[n]$, and
we would like to pick an element from $[n] \setminus
I$ uniformly at random. 
We store $I$
in a binary search tree $T$. We maintain the invariant
that $T$ stores with each element $x \in [n-k+1] \cap I$ 
a \emph{replacement} $\rho_x \in \{n-k+2, \dots, n\} \setminus I$ such 
that $[n] \setminus I = ([n-k+1] \setminus I) \cup
\{\rho_x \mid  x \in [n-k+1] \cap I\}$, see
Figure~\ref{fig:replacement}.
\begin{figure}[htbp]
\centering
\includegraphics{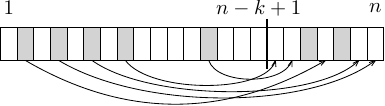}
\caption{Sampling a random sequence $I$ from
$[n]$. At the beginning of round $k$, we
have already sampled $k-1$ elements (shown in gray).
Each element $x \in [n-k+1] \cap I$ 
has a \emph{replacement} $\rho_x \in \{n-k+2, \dots, n\} \setminus I$
(indicated by the arrows).
In round $k$, we pick a random number $x \in [n-k+1]$.
If $x$ is already contained in $I$, we add $\rho_x$ to
$I$. Otherwise, we add $x$.
}
\label{fig:replacement}
\end{figure}
In round $k$, we sample a random number $x$ from
$[n - k + 1]$, and we check in $T$ whether
$x \in I$. If not, we add $x$ to $I$ (and $T$), otherwise, we
add $\rho_x$ to $I$ (and $T$).
By the invariant, we add a uniform random element from $[n] \setminus I$
to $I$.

It remains to update the replacements, see Figure~\ref{fig:rep_cases}.
If $x = n - k + 1$, we do not need a replacement for $x$. 
Now suppose $x < n - k + 1$.
If $n - k + 1 \not\in I$,
we set $\rho_x = n - k + 1$.
Otherwise, we set $\rho_x = \rho_{n-k+1}$.
This ensures that the invariant holds at the beginning
of round $k+1$, and it
takes $O(\log s)$ time and $O(s)$ space. We continue for $s$ rounds.
At the end of the first phase,
any sequence of $s$ distinct numbers in
$[n]$ is sampled with equal probability. Furthermore,
the phase takes $O(s \log s)$ time and $O(s)$ space.
\begin{figure}[htbp]
\centering
\includegraphics{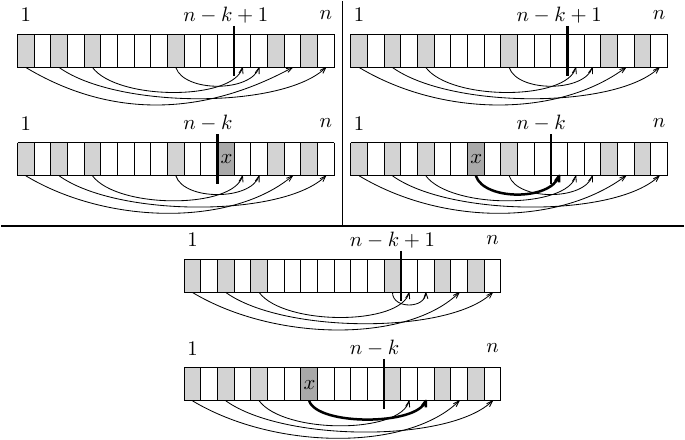}
\caption{Finding a replacement for $x$.
If $x = n - k + 1$, we do not need a replacement for $x$
in the next round (top left).
If $n-k+1$ is not sampled yet, we can make it the replacement for $x$
(top right). Otherwise, we make the old replacement for $n-k+1$ the
new replacement for $x$ (bottom).
}
\label{fig:rep_cases}
\end{figure}

In the second phase,
we scan through $S$ to obtain the elements
whose positions correspond to the
numbers in $I$. This requires $O(n)$ time and $O(s)$ space.
\end{proof}

We use Lemma~\ref{lem:sampling} to find a random sample
$R \subseteq S$ of size $s$. We compute
$\VD(R \cup K)$, triangulate the bounded cells and construct
a planar point location structure for the triangulation.
This takes $O(s \log s)$ time and $O(s)$ space~\cite{Kirkpatrick83}.
By our choice of $K$, all Voronoi cells for points in $R$
are bounded, and every point in $S$ lies in a bounded Voronoi
cell.
Given a vertex $v \in \VD(R \cup K)$, the \emph{conflict
circle} of $v$ is the largest circle with center $v$ and
no point from $R \cup K$ in its interior.
The \emph{conflict set} $B_v$ of $v$ contains all points from $S$ that
lie in the conflict circle of $v$, and the \emph{conflict size}
$b_v$ of $v$ is $|B_v|$.
We scan through $S$
to find the conflict size $b_v$ for each vertex
$v \in \VD(R \cup K)$:
every Voronoi vertex has a
counter that is initially $0$.
For each $p \in S \setminus (R \cup K)$, we use the point location structure to
find the triangle $\Delta$ of $\VD(R \cup K)$ that
contains it. At least one vertex $v$ of $\Delta$ is in conflict with $p$.
Starting from $v$, we walk along the edges of
$\VD(R \cup K)$ to find all Voronoi vertices in conflict with $p$
(recall that these vertices induce a connected component in $\VD(R \cup K)$).
We increment the counters of all these vertices.
This may take a long time in the worst case, so we
impose an upper bound on the total work.
For this, we choose a \emph{threshold} $M$.
When the sum of the conflict
counters exceeds $M$, we start over with a new
sample $R$. The total time for one attempt is
$O(n\log s + M)$, and below we prove that for
$M = \Theta(n)$, the success probability
is at least $3/4$.
Next, we pick another threshold $T$, and we compute
for each vertex $v$ of $\VD(R \cup K)$ the \emph{excess}
$t_v = b_v s/n$. The excess measures how far the vertex
deviates from the desired conflict size $n/s$. We check if
$\sum_{v \in \VD(R \cup K)} t_v \log t_v \leq T$.
If not, we start over with a new sample. Below, we prove that
for $T = \Theta(s)$, the success probability is
at least $3/4$. The total success probability is
$1/2$, and the expected number of attempts is $2$. Thus,
in expected time $O(n\log s + s\log  s)$, we can find a sample
$R \subseteq S$ with
$\sum_{v \in \VD(R \cup K)} b_v = O(n)$
and
$\sum_{v \in \VD(R \cup K)} t_v \log t_v = O(s)$.

We now analyze the success probabilities, using
the classic Clarkson-Shor method~\cite{ClarksonSh89}.
We begin with a variant of the
Chazelle-Friedman bound~\cite{ChazelleFr90}.

\begin{lemma}\label{lem:chazelle_friedman}
Let $X$ be a planar point set of size $m$, and let
$Y \subset \R^2$ with
$|Y| \leq 3$ and $X \cap Y = \emptyset$.  For fixed $p \in (0,1]$, let
$R \subseteq X$ be a random subset of size $pm$ and let
$R' \subseteq X$ be a random subset of size $p'm$, for
$p' = p/2$. Suppose that $p'm \geq 4$. Fix $\uu \subset X \cup Y$
with $|\uu| = 3$, and let
$v_\uu$ be the Voronoi vertex defined by $\uu$. Let $b_\uu$ be the
number of points from $X \cup Y$ in the interior of the circle 
with center $v_\uu$
and with the points from $\uu$ on the boundary. 
Then,
\[
\Pr[v_\uu \in \VD(R \cup Y)]
\leq
64 e^{-pb_\uu/2} \Pr[v_\uu \in \VD(R' \cup Y)].
\]
\end{lemma}

\begin{proof}
Let $\sigma = \Pr[v_\uu \in \VD(R \cup Y)]$ and
$\sigma' = \Pr[v_\uu \in \VD(R' \cup Y)]$.
The vertex $v_\uu$ is in $\VD(R \cup Y)$ precisely
if $\uu \subseteq R \cup Y$ and $B_\uu \cap (R \cup Y) = \emptyset$,
where $B_\uu$ are the points from $X \cup Y$ inside the circle with
center $v_\uu$ and with the points from $\uu$ on the boundary.  
If $B_\uu \cap Y \neq \emptyset$, then
$\sigma = \sigma' = 0$, and the lemma holds. Thus, assume that
$B_\uu \subseteq X$.
Let $d_\uu = |\uu \cap X|$, the number of points in $\uu$ from $X$.
There are $\binom{m - b_\uu - d_\uu}{pm-d_\uu}$ ways
to choose a $pm$-subset from $X$ that avoids all
points in $B_\uu$ and contains all points of $\uu \cap X$,
so

\begin{align*}
\sigma &= \binom{m - b_\uu - d_\uu}{pm - d_\uu}
\left/\binom{m}{pm}\right.\\
&= \frac{\prod_{j=0}^{pm - d_{\uu} - 1}
(m-b_\uu - d_\uu-j)}
{\prod_{j=0}^{pm-d_\uu - 1} (pm - d_\uu - j)}
\left/
  \frac{\prod_{j=0}^{pm-1} (m-j)}
  {\prod_{j=0}^{pm-1} (pm-j)}\right.\\
&= \prod_{j=0}^{d_\uu-1} \frac{pm-j}{m-j} \;\cdot\;
\prod_{j=0}^{pm-d_\uu-1}
\frac{m-b_\uu-d_\uu-j}{m-d_\uu-j}\\
&\leq p^{d_\uu}
\,\prod_{j=0}^{pm-d_\uu-1}
\left(1 -\frac{b_\uu}{m-d_\uu-j}\right).
\end{align*}
Similarly, we get
\[
\sigma' = \prod_{i = 0}^{d_\uu - 1}
\frac{p'm - i}{m -i }\;
\prod_{j = 0}^{p'm - d_\uu - 1}
\left(1 - \frac{b_\uu}{m - d_\uu-j}\right),
\]
and since $p'm \geq 4$ and $i \leq 2$, it
follows that
\[
\sigma'  \geq
\left(\frac{p'}{2}\right)^{d_\uu}\;
\prod_{j=0}^{p'm-d_\uu-1}
\left(1 -\frac{b_\uu}{m - d_\uu - j}\right).
\]
Therefore, since $p' = p/2$,
\[
\frac{\sigma}{\sigma'} \leq \left(\frac{2p}{p'}\right)^{d_\uu}
\prod_{j = p'm - d_\uu}^{pm - d_\uu-1}
\left(1 - \frac{b_\uu}{m - d_\uu - j}\right)\\
\leq 64 \left(1 - \frac{b_\uu}{m}\right)^{pm / 2} \leq
64\,e^{-pb_\uu / 2}.
\]
\end{proof}
We can now bound the total expected conflict size.
\begin{lemma}\label{lem:conflict_size}
We have
$
\Ex\left[\sum_{v \in \VD(R \cup K)}b_v\right]
= O(n)$.
\end{lemma}
\begin{proof}
By expanding the expectation, we get
\begin{align*}
\Ex\left[\sum_{v \in \VD(R \cup K)}b_v\right] &=
\sum_{\uu \subset S \cup K, |\uu| = 3}\Pr[v_\uu \in \VD(R \cup K)] b_\uu,\\
\intertext{with $v_\uu$ being the Voronoi vertex of $\uu$ and
$b_\uu$ its conflict size.
By Lemma~\ref{lem:chazelle_friedman} with
$X = S$, $Y = K$ and $p = s/n$, this is}
    &\leq
\sum_{\uu \subset S \cup K, |\uu| = 3}
64 e^{-pb_\uu/2}
\Pr[v_\uu
\in \VD(R' \cup K)]b_\uu,\\
\intertext{where $R' \subseteq S$ is a sample
of size $s/2$. We bound this as}
&\leq
    \sum_{i = 0}^{\infty}\sum_{\substack{\uu \subset S \cup K, |\uu| = 3\\
	    b_\uu \in [\frac{i}{p},  \frac{i+1}{p})}}
    \frac{64 e^{-i/2}(i+1)}{p}
    \Pr[v_\uu \in \VD(R' \cup K)]\\
&\leq
    \frac{1}{p}\, \sum_{\uu \subset S \cup K, |\uu| = 3}
    \Pr[v_\uu \in \VD(R' \cup K)]
    \sum_{i = 0}^{\infty}
    64 e^{-i/2}(i+1)\\
&= O(s/p) =  O(n),
  \end{align*}
since $\sum_{\uu \subset S \cup K, |\uu| = 3} \Pr[v_\uu \in \VD(R' \cup K)] = O(s)$
is the size of $\VD(R' \cup K)$ and
$\sum_{i = 0}^{\infty}e^{-i/2}(i+1) = O(1)$.
\end{proof}
\noindent
By Lemma~\ref{lem:conflict_size} and Markov's inequality, we
can conclude
that there is an $M = \Theta(n)$ such that
$\Pr[\sum_{v \in \VD(R \cup K)}b_v > M] \leq 1/4$. 
\begin{lemma}\label{lem:excess}
$\Ex\left[\sum_{v \in \VD(R \cup K)}
        t_v\log t_v\right]
  = O(s)$.
\end{lemma}
\begin{proof}
By Lemma~\ref{lem:chazelle_friedman} with $X = S$, $Y = K$, and
$p = s/n$,
\begin{align*}
\Ex\left[\sum_{v \in \VD(R \cup K)}
        t_v\log t_v\right]
	&=
    \sum_{\uu \subset S \cup K, |\uu| = 3}\Pr[v_\uu \in \VD(R \cup K)]\, t_\uu\log t_\uu\\
    &\leq
\sum_{\uu \subset S \cup K, |\uu| = 3}
64 e^{-pb_\uu/2}
\Pr[v_\uu \in \VD(R' \cup K)]t_\uu\log t_\uu\\
&\leq
    \sum_{i = 0}^{\infty}\sum_{\substack{\uu \subset S \cup K, |\uu| = 3\\
	    b_\uu \in [\frac{i}{p}, \frac{i+1}{p})}}
	    64 e^{-\frac{i}{2}}(i+1)^2
    \Pr[v_\uu \in \VD(R' \cup K)]\\
&\leq
    \sum_{i = 0}^{\infty}
    64 e^{-\frac{i}{2}}(i+1)^2
    \sum_{\uu \subset S \cup K, |\uu| = 3}
    \Pr[v_\uu \in \VD(R' \cup K)]\\
&=
  O(s).
  \end{align*}
\end{proof}

By Markov's inequality and Lemma~\ref{lem:excess},
we can conclude that there is a $T = \Theta(s)$ with
$\Pr[\sum_{v \in \VD(R \cup K)} t_v \log t_v \geq T]
\leq 1/4$. This finishes the first sampling phase.

The next goal is to
sample for each vertex $v$ with
$t_v \geq 2$ a random subset $R_v \subseteq B_v$ of
size $\min\{\alpha t_v \log t_v, b_v\}$ for large enough $\alpha > 0$ 
(recall that
$B_v$ is the conflict set of
$v$ and that $b_v = |B_v|$).
\begin{lemma}
In total time $O(n\log s)$, we can sample for each
vertex $v \in \VD(R \cup K)$ with $t_v \geq 2$ a random
subset $R_v \subseteq B_v$ of size
$\min\{\alpha t_v \log t_v, b_v\}$.
\end{lemma}
\begin{proof}
First, we sample for each vertex $v$
with $t_v \geq 2$ a sequence $I_v$ of
$\alpha t_v\log t_v$ distinct numbers from
$\{1, \dots, b_v\}$. For this, we use the first phase of the algorithm from
the proof of Lemma~\ref{lem:sampling} for each such vertex,
but \emph{without} reusing the space.
As explained in the proof of Lemma~\ref{lem:sampling},
this takes total time
\[
O\left(\sum_v (t_v \log t_v) \log (t_v \log t_v)\right)
  =  O\left(\sum_v (t_v \log t_v) \log s\right)
  = O(s \log s),
\]
since $\sum_v t_v\log t_v = O(s)$, and in particular 
$t_v\log t_v = O(s)$ for each vertex $v$
(note that the constant in the O-notation is independent of $v$).
Also, since $\sum_v t_v\log t_v = O(s)$, the total space 
requirement is $O(s)$.

After that, we scan through $S$. For each vertex
$v$, we have a counter $c_v$, initialized to $0$. For each
$p \in S$, we find the conflict vertices of $p$,
and for each conflict vertex $v$, we increment
$c_v$. If $c_v$ appears in the corresponding set $I_v$, we add
$p$ to $R_v$.
The total running time is $O(n\log s)$, as we do one
point location for each input point and the total conflict
size is $O(n)$.
\end{proof}

We next show that for a \emph{fixed} vertex
$v \in \VD(R \cup K)$, with constant probability, all vertices
in $\VD(R_v)$ have conflict size $n/s$ with respect to
$B_v$.
\begin{lemma}\label{lem:second_sample}
Let $v \in \VD(R \cup K)$ with $t_v \geq 2$,
and let $R_v \subseteq B_v$ be the sample for $v$.
The expected number of vertices $v'$ in $\VD(R_v)$
with at least $n/s$ points from $B_v$ in their conflict circle is
at most $1/4$.
\end{lemma}
\begin{proof}
If $R_v = B_v$, the lemma holds, so assume that
$\alpha t_v \log t_v < b_v$.
Recall that $t_v = b_v s/n$.
We have
\begin{align*}
\Ex\Biggl[\sum_{\substack{v' \in \VD(R_v)\\
          b'_{v'} \geq n/s}}
  1\Biggr] &=
    \sum_{\substack{\uu \subset B_v, |\uu| = 3\\b'_\uu \geq n/s}}
    \Pr[v'_\uu \in \VD(R_v)], \\
    \intertext{where $b'_\uu$ denotes the number of points
    from $B_v$ inside the circle with center $v'_\uu$ and with the
    points from $\uu$ on the boundary.
      Using
Lemma~\ref{lem:chazelle_friedman} with $X = B_v$,
$Y = \emptyset$, and
$p = (\alpha t_v \log t_v)/b_v = \alpha(s/n)\log t_v$, this is
}
    &\leq
\sum_{\substack{\uu \subset B_v, |\uu| = 3\\b'_\uu \geq n/s}}
64 e^{-pb'_\uu/2}
\Pr[v'_\uu
\in \VD(R'_v)]\\
&\leq
64 e^{-(\alpha/2) \log t_v}
\sum_{\uu \subset B_v, |\uu| = 3}
\Pr[v'_\uu
\in \VD(R'_v)]\\
&= O(t_v^{-\alpha/2} t_v \log t_v) \leq 1/4,
  \end{align*}
  for $\alpha$ large enough (remember that $t_v \geq 2$).
\end{proof}
\begin{figure}[htbp]
\centering
\includegraphics[scale=.8]{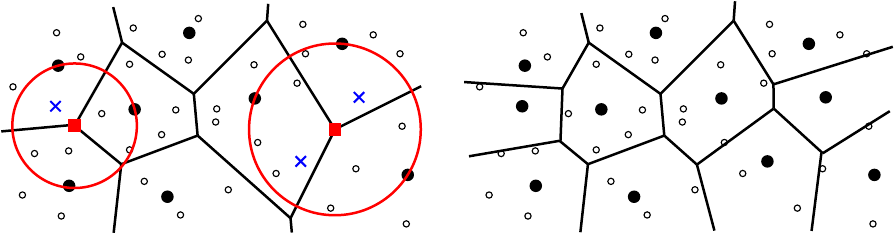}
\caption{A Voronoi Diagram of the sampled set $R$ (left). The two 
red square vertices of $\VD(R \cup K)$ are not good and we need 
to resample within their conflict list (the blue crosses) and compute 
the new Voronoi Diagram (right).}
\label{fig:resampledvd}
\end{figure}
By Lemma~\ref{lem:second_sample} and Markov's inequality, the
probability that all vertices from $\VD(R_v)$ have at most
$n/s$ points from $B_v$ in their conflict circles is at least
$3/4$. If so, we call $v$ \emph{good}, see Figure~\ref{fig:resampledvd}. 
Scanning through $S$,
we can identify the good vertices in time
$O(n\log s)$ and space $O(s)$. Let  $s'$ be the size of $\VD(R \cup K)$.
If we have less than $s'/2$ good vertices, we repeat the process.
Since the expected number of good vertices is $3s'/4$, the
probability that there are at least $s'/2$ good vertices is at least
$1/2$, by Markov's inequality.
Thus, in expectation, we need to perform the sampling twice.
For the remaining
vertices, we repeat the process, but now we take two samples per
vertex, decreasing the
failure probability to $1/4$. We repeat the process, taking
in each round the maximum number of samples that fit into
the work space. In general, if we have $s'/a_i$
active vertices in round $i$, we can take $a_i$ samples per
vertex, resulting in a failure probability of $2^{-a_i}$. Thus,
the expected number of active vertices in round $i+1$ is
$s'/a_{i+1} = s'/(a_i2^{a_i})$. After
$O(\log^* s)$ rounds, all vertices are good.
To summarize:

\begin{lemma}\label{lem:main_sample}
In total expected time $O(n\log s\log^* s)$
and space $O(s)$, we can find
sets $R \subseteq S$ and $R_v \subset B_v$
for each vertex $v \in \VD(R \cup K)$ such that
(i) $|R| = s$: (ii)
$\sum_{v \in \VD(R \cup K)} |R_v| = O(s)$; and (iii)
for every $R_v$, all vertices of $\VD(R_v)$ have
at most $n/s$ points from $B_v$ in their conflict circle.
\end{lemma}
We set $R_2 = R \cup \bigcup_{v \in \VD(R \cup K)} R_v$.
By Lemma~\ref{lem:main_sample}, $|R_2| = O(s)$.
We compute $\VD(R_2 \cup K)$ and triangulate its
bounded cells. For a triangle $\Delta$ of the
triangulation, let $r \in R_2 \cup K$ be
the site whose cell contains $\Delta$, and $v_1, v_2, v_3$
the vertices of $\Delta$.
We set $B_\Delta = \{r\} \cup \bigcup_{i=1}^3 B_{v_i}$. Using the next lemma, we
show that $|B_\Delta| = O(n/s)$.

\begin{lemma}\label{lem:ccirc_cov}
  Let $S\subset\R^2$ and $\Delta=\{v_1,v_2,v_3\}$ a triangle in the
  triangulation of $\VD(S)$. Let $x \in \Delta$. Then any circle $C$
  with center $x$ that contains no points from $S$ is covered by the conflict
  circles of $v_1,v_2$ and $v_3$.
\end{lemma}
\begin{proof}
  Let $p\in C$ and let $r\in S$ be the site whose
  cell contains $\Delta$. We show that $p$ is contained in the
  conflict circle of $v_1$, $v_2$, or $v_3$.
  Consider the bisector $B$ of $p$ and $r$. Since $C$ contains
  $p$ but not $r$, we have $d(x, p) < d(x, r)$, so
  $x$ lies on the same side of $B$ as $p$. Since $x \in \Delta$,
  at least one of $v_1$, $v_2$,
  $v_3$, is on the same side of $B$ as $p$; say $v_1$. This means that
  $d(v_1, p) < d(v_1, r)$, so $p$ lies inside the circle around $v_1$
  with $r$ on the boundary. This is precisely the conflict circle of
  $v_1$.
\end{proof}

\begin{lemma}
  Any triangle $\Delta$ in the triangulation of
  $\VD(R_2\cup K)$ has $|B_\Delta| = O(n/s)$.
\end{lemma}
\begin{proof}
  Let $v$ be a vertex of $\Delta$. We show that $b_v = O(n/s)$.
  Let $\Delta_R=\{v_1,v_2,v_3\}$ be the
  triangle in the triangulation of $\VD(R)$ that contains $v$. By
  Lemma~\ref{lem:ccirc_cov}, we have $B_v \subseteq \bigcup_{i=1}^3 B_{v_i}$.
  We consider the intersections $B_v\cap B_{v_i}$, for $i=1,2,3$.
  If $t_{v_i}<2$, then $b_{v_i}=O(n/s)$ and $|B_v\cap B_{v_i}|=O(n/s)$.
  Otherwise, we have sampled a set $R_{v_i}$ for $v_i$. Let
  $\Delta_{i}=\{w_1,w_2,w_3\}$
  be the triangle in the triangulation of $\VD(R_{v_i})$ that contains $v$.
  Again, by Lemma~\ref{lem:ccirc_cov}, we have $B_v \subseteq \bigcup_{j=1}^3
  B_{w_j}$ and thus also $B_v\cap B_{v_i}\subseteq \bigcup_{j=1}^3
  B_{w_j}\cap B_{v_i}$. However, by construction of $R_{v_i}$, $|B_{w_j}\cap
  B_{v_i}|$ is at most $n/s$ for $j=1,2,3$. Hence, $|B_v\cap B_{v_i}|=O(n/s)$
  and $b_v=O(n/s)$.
\end{proof}

The following lemma enables us to compute the Voronoi diagram of $R_2\cup K$
locally for each triangle $\Delta$ in the triangulation of $\VD(R_2\cup K)$ by
only considering sites in $B_\Delta$. It is a direct consequence of
Lemma~\ref{lem:ccirc_cov}.

\begin{lemma}\label{lem:d_n_c}
For every triangle $\Delta$ in the triangulation of
$\VD(R_2 \cup K)$, we have
$\VD(S \cup K) \cap \Delta = \VD(B_\Delta) \cap \Delta$.
\end{lemma}

\begin{theorem}\label{thm:voronoi}
Let $S$ be a planar $n$-point set.
In expected time $O((n^2/s)\log s + n \log s\log^*s)$
and space $O(s)$, we can compute all Voronoi vertices of $S$.
\end{theorem}

\begin{proof}
We compute a set $R_2$ as above. This takes
$O(n\log s \log^* s)$ time and space $O(s)$.
We triangulate the bounded cells of $\VD(R_2 \cup K)$
and compute a point location structure for the result.
Since there are $O(s)$ triangles, we can store
the resulting triangulation in the workspace.
Now, the goal is to compute simultaneously for
all triangles $\Delta$ the Voronoi diagram $\VD(B_\Delta)$
and to output all Voronoi vertices that lie in $\Delta$
and are defined by points from $S$.
By Lemma~\ref{lem:d_n_c}, this gives all Voronoi vertices of $\VD(S)$.

Given a planar $m$-point set $X$, the algorithm by Asano et
al.~finds all vertices of $\VD(X)$ in
$O(m)$ scans over the input, with constant
workspace~\cite{amrw-cwagp-10}. We can perform a simultaneous
scan for all sets $B_\Delta$ by determining for each point
in $S$ all sets $B_\Delta$ that contain it. This takes total
time $O(n\log s)$, since we need one point location for
each $p \in S$ and since the total size of the $B_\Delta$'s
is $O(n)$.
We need $O(\max_\Delta |B_\Delta|) = O(n/s)$ such scans,
so the second part of the algorithm needs
$O((n^2/s)\log s)$ time.
\end{proof}

As mentioned in the introduction, Theorem~\ref{thm:voronoi} also lets
us report all edges of the Delaunay triangulation of $S$ in the same time bound:
by duality, the three sites that define a vertex of $\VD(S)$ also
define a triangle for the Delaunay triangulation. Thus, whenever we
discover a vertex of $\VD(S)$, we can instead output the corresponding
Delaunay edges, while using a consistent tie-breaking rule to make sure
that every edge is reported only once.

\paragraph{Acknowledgments}
This work began while W.~Mulzer, P.~Seiferth, and Y.~Stein visited the
Tokuyama Laboratory at Tohoku University.
We would like to thank Takeshi Tokuyama and all
members of the lab for their hospitality and for creating a
conducive and stimulating research environment.

\bibliographystyle{abbrv}
\bibliography{visi}

\newpage
\appendix

\section{The Asano-Kirkpatrick Algorithm}
\label{app:ak}

We give more details on the algorithm 
of Asano and Kirkpatrick~\cite{ak-tstanlnp-13}.
Let $F$ be a mountain with vertices 
$q_1, \dots, q_n$ sorted in $x$-order
and base $q_1q_n$. We define the \emph{height} 
$h(q_i)$ of $q_i$, $i=1, \dots, n$, as the
distance from $q_i$ to the line through the base. 
Let $A=(q_1, \dots, q_n)$ be the input array. A 
vertex $q_r$ is the \emph{nearest-smaller-right-neighbor} 
(NSR) of a vertex $q_l$ if (i) $l <r $; 
(ii) $h(q_l) > h(q_r)$; and (iii) $h(q_l) \leq h(q_k)$ 
for $l < k < r$. We call $(q_l, q_r)$ a
\emph{NSR-pair}, with \emph{left endpoint} 
$q_l$ and \emph{right endpoint} $q_r$. 
\emph{Nearest-smaller-left-neighbors} 
(NSL) and NSL-pairs are defined similarly. Let 
$R$ be the set of all NSR-pairs and $L$ be the set of
all NSL pairs. Asano and Kirkpatrick show that the edges $R \cup L$
triangulate $F$. We describe the algorithm for computing $R$. 
The algorithm for $L$ is the same, but it reads the
input in reverse.

Let $s$ denote the space parameter. The algorithm 
runs in $\log_s n$ \emph{rounds}. In
round $i$, $i = 0, \dots, \log_s n - 1$, we 
partition $A$ into $s^i$ consecutive \emph{blocks} 
of size $n/s^i$. Each block $B$ is further 
partitioned into $s$ consecutive \emph{sub-blocks}
$b_{1}, \dots, b_{s}$ of size $n/s^{i+1}$. 
In each round, we compute only NSR-pairs 
with endpoints in different sub-blocks of 
the same block. We handle each block $B$ 
individually as follows. The sub-blocks of $B$ 
are visited from left to right. When we visit a
sub-block $b_{j}$, we compute all NSR-pairs 
with a right endpoint in $b_{j}$ and a
left endpoint in the sub-blocks 
$b_{1},\dots, b_{j-1}$.
Initially, we visit the first sub-block 
$b_{1}$ and we push a pointer to the
rightmost element in $b_{1}$ onto a stack $S$. 
We call a sub-block with a pointer in $S$ 
\emph{active}. Assume now that we have already visited
sub-blocks $b_{1},\dots,b_{j-1}$. 
Let $l$ be the topmost pointer in $S$, referring 
to an element $q_l$ in $b_{j'}$, $j' < j$. 
Furthermore, let $r$ be a pointer to the leftmost 
element $q_r$ in $b_j$.
If $h(q_l) > h(q_r)$, we output $(q_l, q_r)$ and we decrement 
$l$ until we find the first element whose height is smaller 
than the current $h(q_l)$. If $l$ leaves $b_{j'}$, this 
sub-block becomes inactive and we remove $l$ from $S$. We
continue with the new topmost pointer as our new $l$. 
On the other hand, if $h(q_l) \leq h(q_r)$, we increment $r$
by one. We continue until either $r$ leaves $b_j$ or
$S$ becomes empty. Then we
push a pointer to the rightmost element in $b_j$ onto $S$
and proceed to the next sub-block.

In each round, the algorithm reads the 
complete input once in $x$-order. In addition, the 
algorithm reads at most once each active sub-blocks in
reverse order. Note that a sub-block becomes
active only once.

\end{document}